\def\year{2021}\relax
\documentclass[letterpaper]{article} 
\usepackage{aaai21}  
\usepackage{times}  
\usepackage{helvet} 
\usepackage{courier}  
\usepackage[hyphens]{url}  
\usepackage{graphicx} 
\usepackage{natbib}  
\usepackage{caption} 
\usepackage{helvet}  
\usepackage{courier}  
\usepackage[hyphens]{url}
\usepackage{graphicx}  
\usepackage{nicefrac}
\usepackage{booktabs}

\usepackage{amssymb,amsfonts,pifont}

\newcommand{\pref}{\succsim}
	\newcommand{\spref}{\ensuremath{\succ}}

\usepackage{algorithm,algorithmic}

\usepackage{amsmath}
\usepackage{amsthm}
\usepackage{mathtools}
\usepackage{amsfonts}
\usepackage{pgfplots}

\usepackage[switch]{lineno}

\usepackage[normalem]{ulem}

\usepackage{array,xspace,multirow,hhline,tikz,colortbl,tabularx,booktabs,fixltx2e,amsmath,amssymb,amsfonts,amsthm}
\usetikzlibrary{arrows}

	\newtheorem{corollary}{Corollary}
	\newtheorem{remark}{Remark}

	\newtheorem{example}{Example}

\newtheorem{definition}{Definition}

\newtheorem{proposition}{Proposition}

\newcommand{\myOmit}[1]{}

\newcommand{\midd}{\mathbin{:}}

\usepackage{color}

\title{Proportionally Representative Participatory Budgeting\\ with Ordinal Preferences}
\author{Haris Aziz and Barton E. Lee\\UNSW Sydney and Data61 CSIRO\\
    Sydney, Australia\\
       haziz@cse.unsw.edu.au,  barton.e.lee@gmail.com}

\begin{document}

\maketitle

   \begin{abstract}
   Participatory budgeting (PB) is a democratic paradigm 
   whereby voters decide on a set of projects to fund with  a limited budget. We consider PB in a setting where voters report ordinal preferences over projects and have (possibly) asymmetric weights. We propose proportional representation axioms and clarify how they fit into other preference aggregation settings, such as multi-winner voting and approval-based multi-winner voting. As a result of our study, we also discover a new solution concept for approval-based multi-winner voting, which we call Inclusion PSC (IPSC). IPSC is stronger than proportional justified representation (PJR), incomparable to extended justified representation (EJR), and yet compatible with EJR. The well-studied Proportional Approval Voting (PAV)  rule produces a committee that satisfies both EJR and IPSC; however, both these axioms can  also be satisfied by an algorithm that runs in polynomial-time. 
   \end{abstract}

   \section{Introduction}

   Participatory budgeting (PB) provides a grassroots and democratic approach to selecting a set of public projects to fund within a given budget~\citep{AzSh20a}. 
   It has been deployed in several cities all over the globe~\citep{Shah07a}.
   In contrast to standard political elections, PB requires consideration of the (heterogeneous) costs of projects and must respect a budget constraint. When examining PB settings formally, standard voting axioms and methods that ignore budget constraints and differences in each project's cost need to be  reconsidered. In particular, it has been discussed in policy circles that the success of PB partly depends on how well it provides representation to minorities~\citep{BRTK03a}. We take an axiomatic approach to the issue of proportional representation in PB.

   In this paper, we consider PB with \emph{weak ordinal preferences}. Ordinal preferences provide a simple and natural input format whereby participants  rank candidate projects and are allowed to express indifference. A special class of ordinal preferences are
   dichotomous preferences (sometimes referred to  as approval ballots); this input format is  used in most real-world applications  of PB. However, in recent years, some PB applications have shifted to requiring linear order inputs. For example, in the New South Wales state of Australia, participants are asked to provide a partial strict ranking over projects.\footnote{\url{https://mycommunityproject.service.nsw.gov.au}} The PB model we consider encompasses both approval ballots and linear order inputs. 

   In most of the PB settings considered, the participants are assumed to have the same weight. However, in many scenarios, symmetry may be violated. For example, in liquid democracy or proxy voting settings, a voter could be voting on behalf of several voters so may have much more voting weight. Similarly, asymmetric weights may naturally arise if PB is used in   settings where voters have contributed different  amounts to a collective budget {or voters are affected by the PB outcome to different extents.} Therefore, we consider PB where voters may have asymmetric  weights. 

   While there is much discussion on fairness and representation issues in PB, there is a critical need to formalize reasonable axioms to capture these goals. We present two new axioms that relate to the proportional representation axiom, \emph{proportionality for solid coalitions (PSC)},  advocated by Dummett for multi-winner elections~\citep{Dumm84a}. 
{PSC has been referred to as ``a sine qua non for a fair election rule''\citep{Wood94a}
and the essential feature of a voting rule that makes it a system of proportional representation~\citep{Tide95a}.} { We use the key ideas underlying PSC to design new axioms for PB settings.} 
Our axioms provide yardsticks against which existing and new rules and algorithms can be measured. We also provide several justifications for our new axioms. 

   \begin{table*}[h]
   \begin{center}
 	   \scalebox{0.9}{
       \begin{tabular}{  l  l  l }
      \toprule
    & Approval Ballots & Ordinal Prefs \\ \midrule
    Divisible &\citep[e.g.][]{BMS05a} & \citep[e.g.][]{AzSt14a}  \\ \midrule
     Indivisible & \citep[e.g.][]{GKS+19a} & \textbf{This paper}  \\ \bottomrule
       \end{tabular} 
       }
   \end{center}
   \caption{Classification of the literature on fair participatory budgeting with ordinal preferences.}
   \label{table:lit}
   \end{table*}

   \paragraph{Contributions}

   We formalize the setting of PB with weak ordinal preferences. Previously, only restricted versions of the setting, such as PB with approval ballots, have been axiomatically studied~\citep{ALT18a}. We then propose two new axioms \emph{Inclusive PSC (IPSC)} and \emph{Comparative PSC (CPSC)} that are meaningful proportional representation and fairness axioms for PB with ordinal preferences. In contrast to previous fairness axioms for PB {with} approval ballots~\citep[{see, e.g.,}][]{ALT18a}, both IPSC and CPSC imply exhaustiveness ({i.e.,} no additional candidate can be funded without exceeding the budget limit).

   We show that an outcome satisfying Inclusive PSC is always guaranteed to exist and can be computed in polynomial time. The concept appears to be the ``right'' concept for several reasons. 
  {First,} 
  it is stronger than the \linebreak {local-BPJR-L} concept proposed for PB when voters have dichotomous preferences~\citep{ALT18a}.  {Second,} it is also stronger than \emph{generalised PSC} for multi-winner voting with ordinal preferences~\citep{AzLe19a}. {Third,} {when voters have dichotomous preferences,}{ it implies the well-studied concept PJR for multi-winner voting,  is incomparable to the EJR axiom~\citep{AEH+18}}{, and yet is compatible with EJR.} {In particular, the well-studied proportional approval voting rule (PAV) computes an outcome that satisfies both IPSC and EJR; however, there also exists polynomial-time algorithms that can achieve this.} Even for this restricted setting, it  is of independent interest. 
   To show that there exists a polynomial-time algorithm to compute an outcome satisfying IPSC, we present the \emph{PB Expanding Approvals Rule (PB-EAR)} algorithm.

   We also show that the CPSC is equivalent to the generalised PSC axiom for multi-winner voting with weak preferences, to Dummett's PSC axiom for multi-winner voting with strict preferences, and to PJR for multi-winner voting with dichotomous preferences.

   \section{Related Work}

   							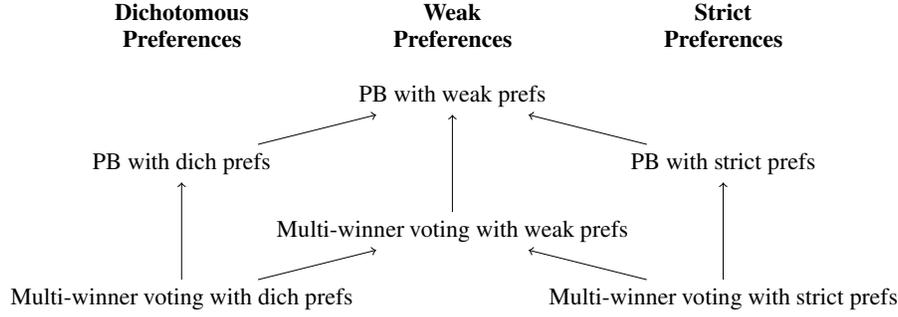
\begin{figure*}[t]
	   								\begin{center}
	  \scalebox{0.9}{ 						\begin{tikzpicture}
	   							\tikzstyle{pfeil}=[->,>=angle 60, shorten >=1pt,draw]
	   							\tikzstyle{onlytext}=[]

	   \node        (weak) at (0,10) {\begin{tabular}{c} \textbf{Weak}\\ \textbf{Preferences} \end{tabular}};

	   \node        (strict) at (4,10) {\begin{tabular}{c} \textbf{Strict}\\ \textbf{Preferences} \end{tabular}};

	   \node        (strict) at (-4,10) {\begin{tabular}{c} \textbf{Dichotomous}\\ \textbf{Preferences} \end{tabular}};


	   			\node        (PBordinal) at (0,9) {PB with weak prefs};

	   						\node        (PBdich) at (-4,8) {PB with dich prefs};

	   						\node        (PBstrict) at (4,8) {PB with strict prefs};

	   				\node        (MWVordinal) at (0,7) {Multi-winner voting with weak prefs};

	   							\node        (MWVdich) at (-4,6) {Multi-winner voting with dich prefs};

	   							\node        (MWVstrict) at (4,6) {Multi-winner voting with strict prefs};

	   		\draw[->] (PBdich) -- (PBordinal);
	   		\draw[->] (PBstrict) -- (PBordinal);
	   		\draw[->] (MWVordinal) -- (PBordinal);
	   		\draw[->] (MWVdich) -- (MWVordinal);
	   		\draw[->] (MWVstrict) -- (MWVordinal);

	   		\draw[->] (MWVdich) -- (PBdich);

	   		\draw[->] (MWVstrict) -- (PBstrict);

	   	%
	   	%
	   	%
	   	%
	   	%

	   						\end{tikzpicture}
						}
	   						\end{center}
	   						\caption{Relations between properties. 
	   						 An arrow from (A) to (B) denotes that (B) is more general than (A).}
	   						 \label{fig:rel}
	   						\end{figure*}

   PB with ordinal preferences can be classified across different axes.
   One axis concerns the input format. Voters either express dichotomous preferences or general weak or linear orders. Along another axis, either the projects are divisible or indivisible. When the inputs are dichotomous preferences, there has been work both for divisible~\citep{BMS05a,ABM19a} as well as indivisible projects~\citep{ALT18a,FT19}. When the input concerns rankings, then there is work where the projects are divisible~\citep[{see, e.g.,}][]{AzSt14a,AAC+19a}. Some of the work is cast in the context of probabilistic voting but is mathematically equivalent to PB for divisible projects. 

 To the best of our knowledge, fairness axioms for PB for discrete projects have not been studied deeply when the input preferences are general ordinal preferences. Therefore, this paper addresses an important gap in the literature.    Table~\ref{table:lit} provides a classification of the literature.

{\citet{ALT18a}, \citet{FT19}, and \citet{BBS20a} focused on PB with discrete projects where the input preference format is approval ballots. We show that our general axioms have connections with proportional representation axioms proposed by \citet{ALT18a} for the case of approval-ballots. We will also show how our approach has additional merit even for the case of approval-ballots. For example, in contrast to previously proposed axioms in~\citet{ALT18a}, our axioms imply a natural property called exhaustiveness.}

   \citet{FST+17a} consider the discrete PB model and study the computational complexity of maximizing various notions of social welfare, including Nash social welfare. {\citet{BNSP17} study issues surrounding preference elicitation in PB with the goal of maximizing utilitarian welfare.}  {In their model, they also consider input formats in which voters express ordinal rankings. However, their focus is not on proportional representation.} 
   \citet{FGM16a} considered PB both for divisible settings as well as discrete settings. However, their focus was on cardinal utilities. In particular, they focus on a demanding but cardinal-utility centric concept of core fairness. 
   Our ordinal approach caters to many settings in which voters only express rankings over projects. Other works on cardinal utilities include~\citet{FMS18a} {and}~\citet{BDG18a}.    In recent work, \citet{REH20a} study an end-to-end model of participatory budgeting and focus primarily on strategic behaviour.


   The paper is also related to a rapidly growing literature on multi-winner voting \citep{ABC+16a,FSST17a,AEF+17b,EFSS17a,Jans16a,Schu02a,Tide06a}. 
   PB is a strict generalization of multi-winner voting.
   Our axiomatic approach is inspired by the PSC axiom in multi-winner voting. The axiom was advocated by \citet{Dumm84a}.
   PSC has been referred to as the most important requirement for proportional representation in multi-winner voting~\citep{Wood94a,Wood97a,TiRi00a,Wood94a,Tide95a}.
 Figure~\ref{fig:rel} provides an overview of which model reduces to which other model.
   We dedicate a separate section to multi-winner voting because one of our axioms gives rise to a new and interesting axiom for the restricted setting of multi-winner voting.



   \section{Preliminaries}


   A PB setting is a tuple $(N,C,\pref,b,w,L)$ where $N$ is the set of $n$ voters, $C$ is the set of candidate projects (candidates), and $L$ is the total \emph{budget limit}.  In the context of PB, it makes sense to refer to $C$ as the set of projects. However, we will also refer to them as candidates especially when making connections with multi-winner voting.
 The function $w: C \rightarrow \mathbb{R^+}$ specifies the cost $w(c)$ of each candidate $c\in C$. We will more generally refer to $w(c)$ as the weight of the candidate project $c$. 
   The function $b: N \rightarrow \mathbb{R^+}$ specifies a \emph{voter weight} $b_i$ for each $i\in N$. 
   We assume that $\sum_{i\in N}b_i$ is $|N|$. For any set of voters $S\subseteq N$, we will denote $\sum_{i\in S}b_i$ by $b(S)$. 
Therefore $b(N)=n$.  {Abusing notation slightly,} for any set of candidates $C'\subseteq C$, we will denote $\sum_{c\in C'}w(c)$ by $w(C')$.  {An outcome, denoted by $W$, is a set of candidates.}
   A set of candidates {(or outcome)} $W\subseteq C$ is \emph{feasible} with respect to $L$ if $w(W)\leq L$. 
  {The preference profile $\pref$ specifies for each voter $i\in N$, her ordinal preference relation over $C$. In the terminology of \citet{BNSP17}, the input format can be viewed as `\emph{rank by value}' so that voters rank projects according to how they value them without taking costs into account.}

We write~$a \pref_i b$ to denote that voter~$i$ values candidate~$a$ at least as much as candidate~$b$ and use~$\spref_i$ for the strict part of~$\pref_i$, i.e.,~$a \spref_i b$ if and only if ~$a \pref_i b$ but not~$b \pref_i a$. Finally, $\sim_i$ denotes~$i$'s indifference relation, i.e., $a \sim_i b$ if and only if both~$a \pref_i b$ and~$b \pref_i a$.
   			 			The relation $\pref_i$ results in (non-empty) equivalence classes $E_i^1,E_i^2, \ldots, E_i^{m_i}$ for some $m_i$ such that $a\spref_i a'$ if and only if $a\in E_i^l$ and $a'\in E_i^{l'}$ for some $l<l'$. Often, we will use these equivalence classes to represent the preference relation of a voter as a preference list.
   					If each equivalence class is of size $1$, then the preference will be a called \emph{strict preference}. If for each voter, the number of equivalence classes is at most two, the preferences are referred to as \emph{dichotomous preferences}. When the preferences of the voters are dichotomous, the voters can be seen as approving a subset of voters. In this case, for each voter $i\in N$, the first equivalence class $E_i^1$ is also referred to as {an} \emph{approval ballot} and is denoted by $A_i\subseteq C$. Note that in this special case, where a voter $i$ has dichotomous preferences, the approval set $A_i$ contains all information about voter $i$'s preference. The vector $A=(A_1,\ldots, A_n)$ is referred to as the \emph{approval ballot profile}. If a voter is indifferent between all candidates, then voter $i$'s approval ballot could be interpreted to be either $A_i=\emptyset$ or $A_i=C$; our results and axioms are independent of this interpretation. 

   Multi-winner voting can be viewed a{s a} special kind of PB {setting} in which $w(c)=1$ for all $c\in C$ and $b_i=1$ for all $i\in N$. The budget limit $L$ is typically denoted by committee size $k$. Any setting that allows for weak preferences can be viewed as encapsulating the corresponding setting with approval ballots. The reason is that approval ballots can be viewed as dichotomous preferences. 

{It will be useful to distinguish between two types of PB outcomes: exhaustive and maximal cost outcomes. These concepts do not rely on the preferences of voters and, instead, are defined solely in terms of the cost of candidates, $w(c)$, and the budget, $L$.}

   \begin{definition}[Exhaustive outcomes]
   An outcome $W$ is said to be exhaustive w.r.t. $L$ if $w(W)\le L$ and $w(W\cup \{c\})>L$ for all $c\in C\backslash W$.
   \end{definition}

   \begin{definition}[Maximal cost outcomes]
   An outcome $W$ is said to be a maximal cost outcome w.r.t. $L$ if  {$W\in \arg\max_{C'}\{ w(C') \ : \ C'\subseteq C \text{ and } w(C')\le L\}$.}
   \end{definition}

   Note that  a maximal cost outcome is always exhaustive but an exhaustive outcome need not be maximal cost. In multi-winner voting, since we only consider outcomes that use up the budget limit of $k$, it means that all feasible outcomes  are both exhaustive and maximal cost. 


   \section{Proportional Representation in PB with Ordinal Preferences}\label{Section: Generalised PSC axioms}

{Before we develop and formally define our concepts, we give some simple examples to provide intuition behind our main ideas.}

{We first warm up with an example that captures the  \emph{proportionality for solid coalitions (PSC)} concept of \citet{Dumm84a}. The example concerns a context in which multi-winner voting coincides with PB.}
\begin{example}[Motivating example I]
{	Suppose there are 9 voters and 4 projects: $a, b, c, d$. The budget limit is 3 dollars and each project costs 1 dollar.
Hence, three projects are to be  selected.  Suppose the preferences of the voters are as follows.
	\begin{align*}
	1-6:& \quad a \succ b \succ c \succ d \ \\
	7-8:&\quad  d \succ c \succ b \succ a\\
	9:&\quad  c \succ a \succ b \succ d
	\end{align*}
	PSC requires that both $a$ and $b$ are selected among the three selected projects.} {Informally speaking,} {the} {rationale} {is that two-thirds of the voters most prefer $a$ and then $b$, and if they are assumed to have control over two-thirds of the budget, then they have the ability to afford both $a$ and $b$.}
\end{example}

 {Following the original PSC axiom for multi-winner elections, our concepts  are based on the idea that if a group of voters is large, and cohesively most prefers a certain set of projects,\footnote{I.e., there is a set of projects that all voters of the group unanimously prefer to all other projects; as will be shown below, this does not require voters to have perfectly aligned preferences.} then sufficient funding should be given to projects within the set. }

\begin{example}[Motivating example II]
{Let voter preferences be }
	\label{example:1}
	\begin{align*}
	1-30:& \quad a \succ b \succ c \succ d\\
	31-100:&\quad  d \succ c \succ b \succ a.
	\end{align*}
 Suppose the total budget limit is $100$, and the weights of the projects are  $w(a)=50$, $w(b)=30$, $w(c)=30$, $w(d)=40$.
The first group of voters (1-30) have $30/100$ of the voter population size. {Our  concepts can be motivated by supposing that  all the voters have equal control of the budget. Thus, the first group of voters can be viewed as controlling 30 units of the total budget limit of 100. However, these voters cannot ``afford" their most preferred project{,} $a${,} as its weight of 50 is more than 30 units of the budget that they control. Yet, the first group of voters' second most preferred project{,} $b${,} is affordable, having weight of only 30 units. Accordingly, the first group of voters can be thought of as having a justified demand that a project no worse than their second-most preferred project is selected, i.e., either project $a$ or $b$. By a similar argument, the second group of voters have a justified demand that both project $d$ and $c$ are selected, since  $w(c)+w(d)\leq 70$. {However, they do not have a justified demand that projects $d, c$ and $b$ are selected, since $w(c)+w(d)+w(b)>70$.}} {Notice that a key difference between multi-winner elections and  the PB setting is that projects may have heterogeneous weights.
}
	\end{example}

The concepts become more complicated when ties are considered in the preference lists.
\begin{example}[Motivating example III]
Consider a modification of Example~\ref{example:1} such that  the  first group of voters are indifferent between $b$ and $c$ as follows.
	\begin{align*}
	1-30:& \quad a \succ b \sim c \succ d,\\
	31-100:&\quad  d \succ c \succ b \succ a.
	\end{align*}
	Then, the voters in the first group would not care if $c$ is selected or {$b$} is selected.
	\end{example}
More generally, our concepts do not require voters in a single group to have perfectly aligned preferences. 

\begin{example}[Motivating example IV]
Consider a modification of Example~\ref{example:1} such that  the  first group of voters are split into two subgroups as follows. 
		\begin{align*}
	1-15:& \quad a \succ b \succ c \succ d\\
	16-30:& \quad b \succ a \succ c \succ d\\
	31-100:&\quad  d \succ c \succ b \succ a.
	\end{align*}
	In this case, the first group (1-15) and the second group (16-30) of voters do not agree on which project is most preferred but they are cohesive in the sense that they unanimously agree that the two-most preferred projects are $a$ and $b$. Since none of the groups can afford their respective most-preferred project with the budget they control,  our concepts require that these two groups are allowed to combine their budgets to make a justified demand for either project $a$ or $b$.
	\end{example}
	
Reasoning about proportional representation become{s}, yet again,  more complicated  when a group of voters can be combined with many different groups. 

\begin{example}[Motivating example V]\label{example: 4}
{Let voter preferences be }
	\begin{align*}
	1-14:& \quad a \succ b \succ c \succ d\\
	15-30:& \quad a \succ c \succ b \succ d  \\
	31-100:&\quad c \succ a \succ b \succ d,
	\end{align*}
	 the total budget limit $100$, and the weights of the projects  $w(a)=90$, $w(b)=30$, $w(c)=80$, $w(d)=40$. Here, the second group of voters (15-30) share a most-preferred project {(project $a$)} with the first group of voters (1-14) but also share their two-most preferred projects{,} $a$ and $c${,} with the third group of voters (31-100). However,  the first and second group  combined  cannot afford project $a$, which has weight 90. Yet, the second and third group can afford project $c$, which has weight $80$.
	\end{example}

	{The last example highlights an additional and key challenge presented by the PB setting that is not present in the multi-winner setting. When groups of voters are combined, their justified demand for projects depends not only on the size of the groups (i.e., the size of the budget that they control), but also the weight of the projects that they prefer. The concepts that we  introduce and develop are flexible enough to capture all of the variants of the example described above. }

{ }{Before presenting our main concepts in the next section, we introduce  the notion of a generalised solid coalition and some technical notation.} {The notion of a generalised solid coalition is central to the PSC of Dummett axioms~\citep{Dumm84a} and the related concepts that we develop.  Intuitively, a set of voters $N'$ forms a generalised solid coalition for a set of candidate projects $C'$ if every voter in $N'$ weakly prefers every candidate project in $C'$ to any candidate project outside of $C'$. Importantly, voters that form a generalised solid coalition for a candidate-project-set $C'$ are not required to have identical preference orderings over candidate projects within $C'$ nor $C\backslash C'$. }

   \begin{definition}[Generalised solid coalition]
   {Suppose voters have weak preferences.} A set of voters $N'$ is a \emph{generalised solid coalition} for a set of candidates $C'$ if every voter in $N'$ weakly prefers each candidate in $C'$ to each candidate in $C\backslash C'$. That is, for all $i\in N'$ and for any $c'\in C'$,
   $\forall c\in C\backslash C' \quad c'\pref_i c. $
   The candidates in $C'$ are said to be \emph{solidly supported} by the voter set $N'${, and conversely the voter set $N'$ is said to \emph{solidly support} the candidate set $C'$.}
   \end{definition}
   
   {Like Dummett's PSC axioms~\citep{Dumm84a}, our axioms will capture intuitive features of proportional representation by ensuring that minority groups of voters are represented in the PB outcome so long as they share similar preferences over candidates, i.e., they form a generalised solid coalition, and the amount of representation given to a group of voters that form a generalised solid coalition is (approximately) in proportion to their size.}

   {Lastly, we introduce some technical notation and terminology that is required for our main concepts.} Let $c^{(i,j)}$ denote voter $i$'s $j$-th most preferred candidate {or one such candidate if indifferences are present}. {To attain such a candidate $c^{(i,j)}$ in the presence of indifferences the following procedure can be used: (1) break all ties in voter $i$'s preferences temporarily to get an artificial strict order and (2) identify the $j$-th candidate $c^{(i,j)}$ in the artificial strict order.
   If a set of voters $N'$ supports a set of candidates $C'$, we will refer to  $\{c\midd \exists i\in N' \text{ s.t. }c \pref_i c^{(i, |C'|)} \} \setminus C'$ as the \emph{periphery} of the set of candidates $C'$ with respect to voter set $N'$.


   \subsection{Main New Concepts}


   %

{
We now present our key concepts for proportional representaton. The concepts are inspired by the PSC concept that was proposed by \citet{Dumm84a} for multi-winner voting for strict preferences. 
The PSC concept requires that if a set of voters $N'$ solidly supports a set of candidates $C'$, then a proportional number of candidates should be selected from $C'$ especially if $C'$ is large enough.} 

{
Although the PSC is quite intuitive and natural, extending it for our general PB settings needs to be done with care. In particular, the presence of candidate weights, budget limits, and indifference cause several complications so we need to define the concepts for the general PB setting carefully. The concepts are based on the requirements put forth on the outcome $W$. Each requirement corresponds to set of voters $N'\subseteq N$ solidly supporting a set of candidates $C'$. Since these voters solidly support $C'$, the proportional representation concepts require that sufficient amount of weight in $W$ should come from either candidates in $C'$ or candidates in the \emph{periphery} of the set of candidates $C'$ with respect to voter set $N'$.\footnote{Allowing for the weight representation to come from the periphery is essential because otherwise even for multi-winner voting, an outcome satisfying the requirements may not exist. }

When formally defining these requirements of the weight composition of $W$, we also need to take care that voters in $N'$ do not require very heavy weight candidates to be included in the outcome. Another guiding principle while formalizing the concepts is that the existence of an outcome satisfying the concepts is not ruled out because of previous insights on subdomains of PB such as multi-winner voting. Next, we use the ideas mentioned above to formally introduce our first key solution concept. }


   \begin{definition}[Comparative PSC (CPSC) for PB with general preferences]\label{definition: CPSC for PB}
   A budget $W$ satisfies \emph{Comparative PSC (CPSC)} if 
   there exists no set of voters $N'\subseteq N$ such that $N'$ solidly supports a set of candidates $C'$ and there is a subset of candidates $C''\subseteq C'$ such that 
   $$w(\{c~\midd~ \exists i\in N' \text{ s.t. }c \pref_i c^{(i, |C'|)} \}\cap W) <w(C'') \le b(N')L/n.$$   
   \end{definition}

   %

   	The intuition for CPSC is that if a set of voters $N'$ solidly supports a subset $C'$ then it may \emph{start to think} that at least weight $b(N')L/n$ worth of candidates should be selected from $C'$ or its periphery especially if there is enough weight present. At the very least it should not be the case that there is a feasible subset of $C''\subseteq C'$ of weight at most $b(N')L/n$ but the weight of $\{c\midd \exists i\in N' \text{ s.t.}c \pref_i c^{(i, |C'|)} \}\cap W$ is strictly less.

   %
   %

   {Inclusion PSC is defined similarly to Comparative PSC.} 

   
   %
      
   \begin{definition}[Inclusion PSC for PB with general preferences]\label{definition: IPSC for PB}
   An outcome $W$ satisfies \emph{Inclusion PSC (IPSC)} if 
   there exists no set of voters $N'\subseteq N$ 
   who have a solidly supported set of candidates $C'$ such that there exists some candidate ${c^*}\in C' \setminus (\{c\midd \exists i\in N' \text{ s.t.}c \pref_i c^{(i, |C'|)} \}\cap W)$ such that 
   $$w({c^*}\cup (\{ c\midd \exists i\in N' \text{ s.t.}c \pref_i c^{(i, |C'|)} \}\cap W))\leq b(N')L/n.$$
%
%
   \end{definition}

   	The intuition for IPSC is that if a set of voters $N'$ solidly supports a subset $C'$ then it may \emph{start to think} that a weight $b(N')L/n$ should be selected from $C'$ or its periphery especially if there is enough weight present. At the very least it should not be the case that weight of $\{c\midd \exists i\in N' \text{ s.t.}c \pref_i c^{(i, |C'|)} \}\cap W$ does not exceed $b(N')L/n$ even if some unselected candidate in ${c^*}\in C'$ can be added to $\{c\midd \exists i\in N' \text{ s.t.}c \pref_i c^{(i, |C'|)} \}\cap W$.

	
   	For both IPSC and CPSC, we avoid violation if for $N'$ solidly supporting candidates in $C'$, the weight of $\{c\midd \exists i\in N' \text{ s.t.}c \pref_i c^{(i, |C'|)} \}\cap W$ is large enough. That is, we only impose representation requirements for sets of voters who solidly support a set of candidates. If, instead, representation requirements were enforced for all sets of voters, regardless of whether they solidly supported a set of candidates or not, then it may not be possible to satisfy either axiom. This observation has already been made in the context of multi-winner voting~{\citep[see, e.g.,][]{ABC+16a}}.  Similarly, both axioms focus on whether the weight $\{c\midd \exists i\in N' \text{ s.t.}c \pref_i c^{(i, |C'|)} \}\cap W$ is large enough. If we only care about the weight of $C'\cap W$, then, again, it can be impossible to satisfy the requirements for all solid coalitions~\citep{ABC+16a}.

	{Next, we show that IPSC and CPSC are independent. The intuition is as follows. CPSC is stronger than IPSC in one respect: it cares about the maximum weight of candidates that are preferred by a coalition of voters whereas IPSC cares about set inclusion. On the other hand, IPSC is stronger in the following respect. For a violation of CPSC, we restrict ourselves to a subset of the solidly supported set of candidates $C''\subseteq C'$. For a violation of IPSC, we even allow for inclusion of a candidate $c$ that is not in the set of solidly supported set of candidates.}

	   		\begin{proposition}\label{prop: CPSC and IPSC differ}
	   			For PB with ordinal preferences, IPSC does not imply CPSC and CPSC does not imply IPSC.
	   			\end{proposition}

   Both IPSC and CPSC imply exhaustiveness as shown in the proposition below.

   \begin{proposition}[CPSC and IPSC are exhaustive]\label{proposition: CPSC and IPSC exhaustive}
   Any outcome $W$ that satisfies CPSC or IPSC is exhaustive.
   \end{proposition}

   %
   %
   %

   {CPSC implies the stronger maximal cost property. As will be shown within the proof of Proposition~\ref{prop: CPSC and IPSC differ},  an IPSC outcome need not be a  maximal cost outcome. }

   \begin{proposition}[CPSC implies maximal cost]\label{proposition: CPSC maximally exhaustive}
   Any outcome $W$ that satisfies CPSC  is a maximal cost outcome.
   \end{proposition}



   \subsection{Concepts with Approval Ballots}

We revisit our central concepts in the special but  well-studied domain of approval ballots. {We provide characterizations of both CPSC and IPSC when the voters have dichotomous preferences. At the end of this section, we show that these characterizations highlight connections between our axioms (CPSC, IPSC) and axioms that have previously been  established in the PB literature.}
{ The following proposition provides a characterization of CPSC in this domain.}

   %
   %

   \begin{proposition}[Comparative PSC (CPSC) for PB with approval preferences]\label{prop: dichotomous pref CPSC}
   Suppose voters have dichotomous preferences. An outcome $W$ satisfies \emph{Comparative PSC (CPSC)} if and only if {the following two conditions hold:}
   \begin{description}
   \item[(i)] there exists no set of voters $N'\subseteq N$ such that there is a subset of candidates $C''\subseteq \bigcap_{i\in N'} A_i$ such that $w(C'')\leq b(N')L/n$ but 
   $w(W\cap \bigcup_{i\in N'}A_i)< w(C'')$, {and}
   \item[(ii)] the outcome $W$ is a maximal cost outcome. 
   \end{description}
   \end{proposition}

   {We also obtain a characterization of IPSC under approval ballots.}

   \begin{proposition}[Inclusion PSC for PB with approval preferences]\label{prop: dichotomous pref IPSC}
   Suppose voters have dichotomous preferences. An outcome $W$ satisfies \emph{Inclusion PSC (IPSC)} if and only if   {the following two conditions hold:}
   \begin{description}
   \item[(i)] there exists no set of voters $N'\subseteq N$ such that
   $w(\cup_{i\in N'}A_i \cap W)< b(N')L/n$ and there exists some $c\in (\cap_{i\in N'}A_i)\setminus (\cup_{i\in N'}A_i\cap W)$ such that $w(\{c\} \cup (\cup_{i\in N'}A_i \cap W))\leq b(N')L/n$, {and}
   \item[(ii)] the outcome $W$ is exhaustive. 
   \end{description}
   \end{proposition}

   PB with approval ballots has been considered by \citet{ALT18a}. For example, they proposed the concept BPJR-L.
 {In the restricted setting studied by \citet{ALT18a},} CPSC for PB with approval preferences is equivalent to the combination of the B-PJR-L and the {\color{black}maximal cost} concepts.
   BPJR-L is weaker than CPSC because BPJR-L does not imply  {\color{black}maximal} cost. 

   %
   %

   \begin{remark}
   In the standard multi-winner setting, outcomes are required to have maximal cost (and hence are exhaustive). Thus, condition (ii) in Proposition~\ref{prop: dichotomous pref CPSC} and~\ref{prop: dichotomous pref IPSC} are always satisfied in the multi-winner setting. 
   \end{remark}
   %
   %
   %
   %
   %

 IPSC for PB with approval preferences is stronger than the Local-BPJR-L proposed by \citet{ALT18a}.

   \section{Computing Proportional Outcomes}

		{In this section, we focus on the computational aspects of proportionally representative outcomes.}
   		Our first observation is that computing a CPSC outcome is computationally hard, even for one voter. The reduction is from the knapsack problem.

   \begin{proposition}
   	Computing a CPSC outcome is weakly NP-hard even for the case of one voter.
   \end{proposition}


   {Next, we show that even for one voter with strict preferences, a CPSC outcome may not exist.}

   \begin{example}
   	Consider the following PB instance with one voter and 4 candidate projects. The voters' preferences are as follows. $	1: a \succ b \succ c \succ d$.  
	The limit $L$ is 4 and the weights are: $w(a)=3, w(b)=w(c)=w(d)=2$. CPSC requires that {project} $a$ must be selected. It also requires that $\{b,c\}$ should be selected. Therefore{,} a CPSC outcome does not exist.
	\end{example}

    Later, we will show that in a more restrictive setting (multi-winner approval voting) a CPSC outcome always exists, can be computed in polynomial-time, and coincides with a well-established proportional representation axiom, called PJR.

    

   %
   %
   %
   %

   							\begin{algorithm}[h!]
   								  \caption{PB Expanding Approvals Rule (PB-EAR)}
   								  \label{algo:newrule}
   							 %
   \normalsize
   								\begin{algorithmic}
									
   									\REQUIRE  $(N,C,\pref,b,L,w)$ \COMMENT{$\pref$ can contain weak {preferences}; if a voter $i$ expresses her preferences over a subset $C'\subset C$, then $C\setminus C'$ is considered the last equivalence class of the voter.}
   									\ENSURE $W\subseteq C$ such that $w(W)\leq L$.
   								\end{algorithmic}
   								\begin{algorithmic}[1]
									
				
   			\STATE $j\longleftarrow 1$; $W \longleftarrow \emptyset$														\WHILE{$w(W)<L$ and no other candidate can be added to $W$ without exceeding budget limit $L$}
   \FOR{$i\in N$}
   			\STATE {$A_i^{(j)}\longleftarrow \{c\in C\, :\, c\pref_i c^{(i,j)}\}$ }
   			\ENDFOR
			
   			\STATE{$C^*\longleftarrow \{c\in C\backslash W\, :\, \sum_{\{i \in N\, :\, c\in A_i^{(j)}\}} b_i \ge n\frac{w(c)}{L}\}$}

   		\IF{{$C^*=\emptyset$}}
   		\STATE $j\longleftarrow j+1$				
   									\ELSE
   									\STATE \label{step-j-approval} 
   									{ Select a candidate $c^*$ from $C^*$ and add it to $W$}
   									\STATE {$N'\longleftarrow \{i\, :\, c^*\in A_i^{(j)}\}$}
   									\STATE \label{step-reweighting} 
   						 Modify the weights of voters in $N'$ {so the total weight of voters in $N'$, i.e., $\sum_{i\in N'} b_i$,} decreases by exactly $n\frac{w(c)}{L}$.  
									
   		%
   		\ENDIF
   	\ENDWHILE
								
   									\RETURN $W$
   								\end{algorithmic}
   							\end{algorithm}

   In contrast to CPSC, we show that an IPSC outcome is not only guaranteed to exist but it can be computed in polynomial time via Algorithm~\ref{algo:newrule}{,} which we refer to as PB-EAR. The algorithm is a careful generalization of the EAR algorithm of \citet{AzLe19a}. 
   In the algorithm, $W$ is initially empty. Some most preferred candidate $c$ is selected (i.e., added into the set $W$) if it has sufficient support $n\cdot(w(c))/L$ from the voters. If $c$ is selected, then  $n\cdot(w(c))/L$ voting weight of the voters who most prefer $c$ is decreased; it does not matter which  of these voters' weight is decreased nor by how much --- so long as a total of $n\cdot(w(c))/L$ voting weight  is reduced. If no such candidate exists, candidates further down in the preference lists of all voters  are considered. 
It is clear that PB-EAR runs in polynomial time. The argument for PB-EAR satisfying IPSC does not depend on what way candidate $c^*$ is selected is Step~\ref{step-j-approval}.


   							\begin{proposition}\label{prop:PBEAR-IPSC}
   							PB-EAR  satisfies  Inclusion PSC for PB. 	
   		\end{proposition}

			We note here that  not all IPSC outcomes are possible outcomes of PB-EAR even for the restricted setting of multi-winner voting.
%


   			%

   %

   \section{Special Focus on Multi-winner Voting}

   In this section, we dive into the well-studied setting of multi-winner voting, which is also referred to as committee voting. In this setting, $k$ candidates are to be selected from the set of candidates. 
Note that PB reduces to multi-winner voting if the weight of each candidate is 1 and the budget limit is set to $k$.  
   
 We uncover some unexpected relations between fairness concepts for this particular setting. We also show that whereas  CPCS does not give rise to a new fairness concept, IPSC gives rise to a new fairness concept even for the setting concerning approval ballots. When discussing concepts for PB, we will assume that voters have equal voter weight of 1. This will make it possible to form connections with concepts for multi-winner voting in which all the voters are typically treated equally. 

   Let us first introduce generalised PSC, which was proposed by \citet{AzLe19a} and applies to multi-winner settings with ordinal preferences.    \citet{AzLe19a} showed that generalised PSC extends the PJR concept for multi-winner voting with approval ballots.

   \begin{definition}[Generalised PSC \citep{AzLe19a}]\label{Definition: Gen PSC}
   A committee $W$ satisfies \emph{generalised PSC} if for every positive integer $\ell$, and for all generalised solid coalitions $N'$ supporting candidate subset $C'$ with size $|N'|\ge \ell n/k$, there exists a set $C^*\subseteq W$ with size at least $\min\{\ell, |C'|\}$ such that for all $c''\in C^*$, $\exists i\in N'\, : \quad c''\pref_i c^{(i, |C'|)}.$
   \end{definition}

   %
   
{In the multi-winner setting, our axioms have connections with previously studied axioms related to PSC. In particular, we show that CPSC is equivalent to generalised PSC, and IPSC implies generalised PSC. The latter result implies that IPSC is a stronger concept than CPSC. This is, perhaps, surprising given that in more general settings  CPSC appears to be a more demanding concept than IPSC because computing a CPSC outcome is NP-hard and a CPSC outcom{e} may may not exist.}


   \begin{proposition}\label{proposition: CPSC equiv to genPSC and IPSC implies genPSC}
   	For multi-winner voting, 
	\begin{description}
	\item[(i)] CPSC is equivalent to Generalised PSC. 
	\item[(ii)]  IPSC implies Generalised PSC (or CPSC) 
	\end{description}
   	\end{proposition}

   As another corollary, we note that since testing PJR is coNP-complete~\citep{AEH+18}, testing CPSC is coNP-complete.

	\subsection{Approval-based multi-winner voting}
	
	
	In this subsection, we explore our axioms in the well-studied setting of approval-based multi-winner elections. We begin by stating two established PR axioms: \emph{Proportional Justified Representation (PJR)}~\citep{SFF+17a} and \emph{Extended Justified Representation (EJR)}~\citep{ABC+16a}.
	
   \begin{definition}[PJR]
   {Suppose all voters have dichotomous preferences.} A committee $W$ with $|W|=k$ satisfies PJR for an approval ballot profile $\boldsymbol{A}=(A_1, \ldots, A_n)$ over a candidate set $C$ if for every positive integer $\ell\le k$ there does not exists a set of voters $N^*\subseteq N$ with $|N^*|\ge \ell \frac{n}{k}$ such that {the following two conditions hold:}
   \begin{description}
   \item[(i)] $\big|\bigcap_{i\in N^*} A_i\big|\ge \ell$, and
   \item[(ii)] $\big|\big(\bigcup_{i\in N^*} A_i\big)\cap W\big|<\ell.$
   \end{description}
   \end{definition}

   		\begin{table*}[h!]
   			\centering
   			\scalebox{1}{
   			\begin{tabular}{lllllll}
   				\toprule
   				\textbf{PB with}&PB with&\textbf{PB} \textbf{with}&Multi-winner&Multi-winner&Multi-winner with\\
   								\textbf{Ordinal Prefs}&Approvals&\textbf{Strict Pref}&with Ordinal Prefs &with Approvals &Strict Prefs\\
   				\midrule
   \textbf{CPSC}&BPJR-L$^{(\#)}$&\textbf{CPSC}&generalised PSC$^{(*)}$&PJR$^{(*)}$ &PSC$^{(*)}$\\
   \midrule
   \textbf{IPCS}&\textbf{IPCS} &\textbf{IPCS}& \textbf{IPCS} & \textbf{IPCS}&\textbf{IPCS}\\
   				\bottomrule
   			\end{tabular}
   			}

   			\caption{Equivalent fairness concepts for social choice settings. The concepts and settings in bold are from this paper. $(*)$ implies that, for the given social choice setting, the fairness concept is equivalent to CPSC.   $(\#)$ implies that, for the given social choice setting, the fairness concept combined with the maximal cost property is equivalent to CPSC. }
   			\label{table:summary:PBconcepts}
   		\end{table*}

		   \begin{definition}[EJR]
   {Suppose all voters have dichotomous preferences.} A committee $W$ with $|W|=k$ satisfies EJR for an approval ballot profile $\boldsymbol{A}=(A_1, \ldots, A_n)$ over a candidate set $C$ if for every positive integer $\ell\le k$ there does not exists a set of voters $N^*\subseteq N$ with $|N^*|\ge \ell \frac{n}{k}$ such that {the following two conditions hold:}
   \begin{description}
   \item[(i)] $\big|\bigcap_{i\in N^*} A_i\big|\ge \ell$, and
   \item[(ii)]  $|A_i\cap W|<\ell$  for each $i\in N^*$.
      \end{description}
   \end{definition}

Our first result is a corollary of Proposition~\ref{proposition: CPSC equiv to genPSC and IPSC implies genPSC}. It states that, in the special case of approval-based multi-winner voting, CPSC, PJR and Generalised PSC are all equivalent.
	
	\begin{corollary}\label{corollary: approval voting CPSC and PJR equiv}
   			For multi-winner voting with approval ballot, CPSC, PJR, and Generalised PSC are equivalent. 
   			\end{corollary}
   			\begin{proof}
   				\citet{AzLe19a} proved that{,} for multi-winner voting with approval ballot{,} PJR and generalised PSC are equivalent.\footnote{{Unlike the present paper,} \citeauthor{AzLe19a}'s (\citeyear{AzLe19a})  model  assumes that no voter is indifferent between all candidates; however, this assumption is not required to show the equivalence.} We have shown that{,} for multi-winner voting, CPSC and generalised PSC are equivalent. 
   				\end{proof}

			{Although the focus of the present paper has been on generalising the multi-winner PSC concept of Dummett~\citep{Dumm84a} to the PB setting, Proposition~\ref{prop: IPSC implies PJR for approval-based mw and EJR} provides a surprising discovery in the reverse direction. In the special case of approval-based multi-winner voting, IPSC is a new PSC axiom that is closely related --- albeit stronger --- than PJR.   In recent years, PJR and its related axioms have been intensely studied by the computational social choice community~\citep[{see, e.g.,}][]{AEH+18,ABC+16a,FSST17a,AEF+17b,EFSS17a}. Given Proposition~\ref{prop: IPSC implies PJR for approval-based mw and EJR} and the community's interest in PJR-like axioms, we  formalise the IPSC axiom for the approval-based multi-winner setting. We also establish a number of results that illustrate the connection between IPSC and other axioms such as PJR, EJR and PAV.}
			
   \begin{proposition}[Inclusion PSC (IPSC)  for multi-winner voting with approval preferences]\label{def: dichotomous pref  multiwinner IPSC}
Suppose voters have dichotomous preferences. A committee $W$ of size $k$ satisfies \emph{Inclusion PSC (IPSC)} if and only if  there exists no set of voters $N'\subseteq N$ such that {the following two conditions hold:}
\begin{description}
\item[(i)]    $|N'|\ge  (|\cup_{i\in N'}A_i \cap W|+1 ) n/k$, and
\item[(ii)]  there exists some $c^*\in \cap_{i\in N'}A_i \backslash (\cup_{i\in N'}A_i\cap W)$.
\end{description}
   \end{proposition}
		
		\begin{proof}
		{Follows from Proposition~\ref{prop: dichotomous pref IPSC}
 by setting $b(N')=|N'|$, $L=k$, $w(C')=|C'|$ for all $C'\subseteq C$, and simplifying.}
		\end{proof}
		
   		 %

   \begin{proposition}\label{prop: IPSC implies PJR for approval-based mw and EJR}
   	For multi-winner voting with approvals, 
	\begin{description}
	\item[(i)] IPSC implies PJR, 
	\item[(ii)]  PJR does not imply IPSC.
	\item[(iii)] IPSC and EJR are incomparable
	\end{description} 
   	\end{proposition}

   We next show that the well-studied voting rule \emph{Proportional Approval Voting (PAV)} produces a committee that satisfies IPSC. Under PAV, a voter $i$ that has $j$ of their approved candidates elected, i.e., $j=|W\cap A_i|$, is assumed to  attain  utility
 $r(j)=\sum_{p=1}^j \frac{1}{p}$ if $j>0$  and $0$ otherwise. 
   
   Given an outcome $W$, the PAV-score of $W$ is the sum of {voter}
    utilities, i.e., $\sum_{i\in N} r(|A_i\cap W|)$.    The output of PAV is an outcome {$W^*$} that has maximal PAV-score, i.e.,  
    $W^*\in \arg\max\{\sum_{i\in N} r(|A_i\cap W|) \ : \ W\subseteq C \text{ and } |W|=k\}$.

		\begin{proposition}\label{prop:pav-implies-ipsc}
PAV satisfies IPSC.
		\end{proposition}
		
%
%
		
   Given that PAV implies IPSC, the above proposition shows that EJR and IPSC are compatible axioms. This follows immediately from combining the above result with the fact that PAV also implies EJR~\citep{ABC+16a}; { however, IPSC and EJR do not characterize PAV. That is, there exists committees that satisfy both EJR and IPSC but are not PAV.}
   

   
     \begin{proposition}\label{Proposition: IPSC and EJR via PAV}
     {\color{white} a}
     \begin{description}
     \item[(i)]   IPSC and EJR are compatible. That is, there always exists a committee outcome  that satisfies both IPSC and EJR. In particular, the output of the PAV rule is such a committee.
     \item[(ii)] A committee satisfying both EJR and IPSC need not be a PAV outcome. 
     \end{description}   
   \end{proposition}



%
%


Part (ii) of  Proposition~\ref{Proposition: IPSC and EJR via PAV} is a double-edged sword. On   one hand, IPSC and EJR are insufficient in characterizing PAV. On the other hand, since PAV is computationally intractable, it suggests that computing committee outcomes that satisfy both axioms may be computationally tractable. Indeed, the following proposition proves that an outcome satisfying both axioms can be computed in polynomial-time. Interestingly, the algorithm that produces this outcome is a special case of the EAR algorithm~\citep{AzLe19a} applied to dichotomous preferences. The algorithm in question is studied by \citet{PeSk20} who call it ``Rule-X.''

\begin{proposition}
A committee satisfying both EJR and IPSC can be computed via a polynomial-time algorithm. 
\end{proposition}


Finally, we conclude by noting that there is no ranking that can be applied to EJR and IPSC in terms of PAV scores. That is, there are instances where an IPSC outcome provides higher PAV-score than an EJR outcome and vice-versa.

To summarize the results of this subsection, we provide a schematic illustration of the relationship between our axioms, PJR, EJR and PAV in Figure~\ref{fig:approvalaxioms}.


\begin{center}
			\begin{figure}[H]
				\centering
			\begin{tikzpicture}[scale=0.50][thick]
   \draw (0,0) circle (2.5cm);
   \draw (3,0) circle (2.5cm);
      \draw (1.5,0.6) circle (4.25cm);
         \draw (1.5,3) node {PJR$\equiv$CPSC$\equiv$Gen-PSC};
   \draw (-1,0) node {EJR};
   \draw (4,0) node {IPSC};
         \draw (1.5,0) circle (0.8cm);
            \draw (1.5,0) node {PAV};
\end{tikzpicture}
\caption{Schematic illustration of PJR, EJR, IPSC, CPSC and PAV for the approval-based multi-winner setting.}
\label{fig:approvalaxioms}
\end{figure}
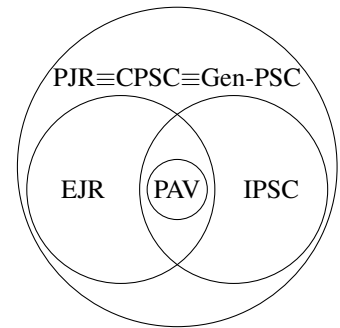
\end{center}

      		\section{Conclusion} 
			
In this paper, we undertook a formal study of PB with ordinal preferences.  Table~\ref{table:summary:PBconcepts} summarizes how some of the concepts are equivalent to each other in particular settings. 
		    We propose two axioms (CPSC and IPSC) that capture  important aspects of the proportional representation. One of the concepts (IPSC) leads to a new  concept even for the restricted setting of multi-winner voting.
{If voters have additive separable utilities over projects, the cardinal utility information can be used to derive the underlying ordinal preferences. Therefore, our axioms and rules also apply to settings where voters have additive separable utilities over projects. 			The formal study of PB from a (computational) social choice perspective is still in its infancy. We envisage further work on axioms and algorithms for fair participatory budgeting.}

\bibliographystyle{aaai}
 

\bibliography{../../adtbib/abb.bib,../../adtbib/adt.bib,../../adtbib/aziz.bib}

\begin{thebibliography}{33}
\providecommand{\natexlab}[1]{#1}
\providecommand{\url}[1]{\texttt{#1}}
\providecommand{\urlprefix}{URL }
\expandafter\ifx\csname urlstyle\endcsname\relax
  \providecommand{\doi}[1]{doi:\discretionary{}{}{}#1}\else
  \providecommand{\doi}{doi:\discretionary{}{}{}\begingroup
  \urlstyle{rm}\Url}\fi

\bibitem[{Airiau et~al.(2019)Airiau, Aziz, Caragiannis, Kruger, Lang, and
  Peters}]{AAC+19a}
Airiau, S.; Aziz, H.; Caragiannis, I.; Kruger, J.; Lang, J.; and Peters, D.
  2019.
\newblock Portioning using Ordinal Preferences: Fairness and Efficiency.
\newblock In \emph{Proceedings of the 28h International Joint Conference on
  Artificial Intelligence (IJCAI)}.

\bibitem[{Aziz, Bogomolnaia, and Moulin(2019)}]{ABM19a}
Aziz, H.; Bogomolnaia, A.; and Moulin, H. 2019.
\newblock Fair Mixing: the Case of Dichotomous Preferences.
\newblock In \emph{Proceedings of the 20th ACM Conference on Electronic
  Commerce (ACM-EC)}, 753--781.

\bibitem[{Aziz et~al.(2017{\natexlab{a}})Aziz, Brill, Conitzer, Elkind,
  Freeman, and Walsh}]{ABC+16a}
Aziz, H.; Brill, M.; Conitzer, V.; Elkind, E.; Freeman, R.; and Walsh, T.
  2017{\natexlab{a}}.
\newblock Justified Representation in Approval-Based Committee Voting.
\newblock \emph{Social Choice and Welfare} 461--485.

\bibitem[{Aziz et~al.(2017{\natexlab{b}})Aziz, Elkind, Faliszewski, Lackner,
  and Skowron:}]{AEF+17b}
Aziz, H.; Elkind, E.; Faliszewski, P.; Lackner, M.; and Skowron:, P.
  2017{\natexlab{b}}.
\newblock The {C}ondorcet Principle for Multiwinner Elections: From
  Shortlisting to Proportionality.
\newblock In \emph{Proceedings of the 26th International Joint Conference on
  Artificial Intelligence (IJCAI)}, 84--90.

\bibitem[{Aziz et~al.(2018)Aziz, Elkind, Huang, Lackner,
  S{\'a}nchez-Fern{\'a}ndez, and Skowron}]{AEH+18}
Aziz, H.; Elkind, E.; Huang, S.; Lackner, M.; S{\'a}nchez-Fern{\'a}ndez, L.;
  and Skowron, P. 2018.
\newblock On the complexity of {E}xtended and {P}roportional {J}ustified
  {R}epresentation.
\newblock In \emph{Proceedings of the 32nd AAAI Conference on Artificial
  Intelligence (AAAI)}, 902--909. AAAI Press.

\bibitem[{Aziz and Lee(2020)}]{AzLe19a}
Aziz, H.; and Lee, B.~E. 2020.
\newblock The Expanding Approvals Rule: Improving Proportional Representation
  and Monotonicity.
\newblock \emph{Social Choice and Welfare} 54(1): 1--45.

\bibitem[{Aziz, Lee, and Talmon(2018)}]{ALT18a}
Aziz, H.; Lee, B.~E.; and Talmon, N. 2018.
\newblock Proportionally Representative Participatory Budgeting: Axioms and
  Algorithms.
\newblock In \emph{Proceedings of the 17th International Conference on
  Autonomous Agents and MultiAgent Systems, {AAMAS} 2018, Stockholm, Sweden,
  July 10-15, 2018}, 23--31.

\bibitem[{Aziz and Shah(2020)}]{AzSh20a}
Aziz, H.; and Shah, N. 2020.
\newblock Participatory Budgeting: Models and Approaches.
\newblock In Rudas; and G{\'a}bor, eds., \emph{In Pathways between Social
  Science and Computational Social Science: Theories, Methods and
  Interpretations}. Springer.

\bibitem[{Aziz and Stursberg(2014)}]{AzSt14a}
Aziz, H.; and Stursberg, P. 2014.
\newblock A Generalization of Probabilistic Serial to Randomized Social Choice.
\newblock In \emph{Proceedings of the 28th AAAI Conference on Artificial
  Intelligence (AAAI)}, 559--565. AAAI Press.

\bibitem[{Baumeister, Boes, and Seeger(2020)}]{BBS20a}
Baumeister, D.; Boes, L.; and Seeger, T. 2020.
\newblock Irresolute Approval-based Budgeting.
\newblock In \emph{Proceedings of the 19th International Joint Conference on
  Autonomous Agents and Multi-Agent Systems (AAMAS)}.

\bibitem[{Benad{\`e} et~al.(2017)Benad{\`e}, Nath, Shah, and
  Procaccia}]{BNSP17}
Benad{\`e}, G.; Nath, W.; Shah, N.; and Procaccia, A.~D. 2017.
\newblock Preference Elicitation for Participatory Budgeting.
\newblock In \emph{Proceedings of the 31st AAAI Conference on Artificial
  Intelligence (AAAI)}. AAAI Press.

\bibitem[{Bhaskar, Dani, and Ghosh(2018)}]{BDG18a}
Bhaskar, U.; Dani, V.; and Ghosh, A. 2018.
\newblock Truthful and Near-Optimal Mechanisms for Welfare Maximization in
  Multi-Winner Elections.
\newblock In \emph{Proceedings of the 32nd AAAI Conference on Artificial
  Intelligence (AAAI)}, 925--932.

\bibitem[{Bhatnaga et~al.(2003)Bhatnaga, Rathore, Torres, and
  Kanungo}]{BRTK03a}
Bhatnaga, D.; Rathore, A.; Torres, N.~M.; and Kanungo, P. 2003.
\newblock Participatory Budgeting in {B}razil.
\newblock \emph{World Bank Empowerment Case Studies} .

\bibitem[{Bogomolnaia, Moulin, and Stong(2005)}]{BMS05a}
Bogomolnaia, A.; Moulin, H.; and Stong, R. 2005.
\newblock Collective choice under dichotomous preferences.
\newblock \emph{Journal of Economic Theory} 122(2): 165--184.

\bibitem[{Dummett(1984)}]{Dumm84a}
Dummett, M. 1984.
\newblock \emph{Voting Procedures}.
\newblock Oxford University Press.

\bibitem[{Elkind et~al.(2017)Elkind, Faliszewski, Skowron, and
  Slinko}]{EFSS17a}
Elkind, E.; Faliszewski, P.; Skowron, P.; and Slinko, A. 2017.
\newblock Properties of Multiwinner Voting Rules.
\newblock \emph{Social Choice and Welfare} .

\bibitem[{Fain, Goel, and Munagala(2016)}]{FGM16a}
Fain, B.; Goel, A.; and Munagala, K. 2016.
\newblock The Core of the Participatory Budgeting Problem.
\newblock In \emph{Web and Internet Economics - 12th International Conference,
  {WINE} 2016, Montreal, Canada, December 11-14, 2016, Proceedings}, 384--399.

\bibitem[{Fain, Munagala, and Shah(2018)}]{FMS18a}
Fain, B.; Munagala, K.; and Shah, N. 2018.
\newblock Fair Allocation of Indivisible Public Goods.
\newblock In \emph{Proceedings of the 19th ACM Conference on Electronic
  Commerce (ACM-EC)}, 575--592.

\bibitem[{Faliszewski et~al.(2017)Faliszewski, Skowron, Slinko, and
  Talmon}]{FSST17a}
Faliszewski, P.; Skowron, P.; Slinko, A.; and Talmon, N. 2017.
\newblock Multiwinner Voting: A New Challenge for Social Choice Theory.
\newblock In Endriss, U., ed., \emph{Trends in Computational Social Choice},
  chapter~2.

\bibitem[{Faliszewski and Talmon(2019)}]{FT19}
Faliszewski, P.; and Talmon, N. 2019.
\newblock A Framework for Approval-based Budgeting Methods.
\newblock In \emph{Proceedings of the 33rd AAAI Conference on Artificial
  Intelligence (AAAI)}.

\bibitem[{Fluschnik et~al.(2017)Fluschnik, Skowron, Triphaus, and
  Wilker}]{FST+17a}
Fluschnik, T.; Skowron, P.; Triphaus, M.; and Wilker, K. 2017.
\newblock Fair Knapsack.
\newblock \emph{CoRR} abs/1711.04520.

\bibitem[{Goel et~al.(2019)Goel, K., Sakshuwong, and Aitamurto}]{GKS+19a}
Goel, A.; K., A.~K.; Sakshuwong, S.; and Aitamurto, T. 2019.
\newblock Knapsack Voting for Participatory Budgeting.
\newblock \emph{ACM Transactions on Economics and Computation (TEAC)} 7(2):
  8:1--8:27.
\newblock ISSN 2167-8375.

\bibitem[{Janson(2016)}]{Jans16a}
Janson, S. 2016.
\newblock Phragm{\'e}n's and {T}hiele's election methods.
\newblock Technical Report arXiv:1611.08826 [math.HO], arXiv.org.

\bibitem[{Peters and Skowron(2020)}]{PeSk20}
Peters, D.; and Skowron, P. 2020.
\newblock Proportionality and the limits of Welfarism.
\newblock In \emph{Proceedings of the 21st ACM Conference on Economics and
  Computation (EC'20)}.

\bibitem[{Rey, Endriss, and de~Haan(2020)}]{REH20a}
Rey, S.; Endriss, U.; and de~Haan, R. 2020.
\newblock Shortlisting Rules and Incentives in an End-to-End Model for
  Participatory Budgeting.
\newblock \emph{CoRR} abs/2010.10309.
\newblock \urlprefix\url{https://arxiv.org/abs/2010.10309}.

\bibitem[{S{\'a}nchez-Fern{\'a}ndez et~al.(2017)S{\'a}nchez-Fern{\'a}ndez,
  Elkind, Lackner, Fern{\'a}ndez, Fisteus, {Basanta Val}, and
  Skowron}]{SFF+17a}
S{\'a}nchez-Fern{\'a}ndez, L.; Elkind, E.; Lackner, M.; Fern{\'a}ndez, N.;
  Fisteus, J.~A.; {Basanta Val}, P.; and Skowron, P. 2017.
\newblock Proportional justified representation.
\newblock In \emph{Proceedings of the 31st AAAI Conference on Artificial
  Intelligence (AAAI)}. AAAI Press.

\bibitem[{Schulze(2002)}]{Schu02a}
Schulze, M. 2002.
\newblock {On Dummett's `Quota Borda System'}.
\newblock \emph{Voting matters} 15(3).

\bibitem[{Shah(2007)}]{Shah07a}
Shah, A. 2007.
\newblock \emph{Participatory Budgeting}.
\newblock Public sector governance and accountability series. The World Bank.

\bibitem[{Tideman(1995)}]{Tide95a}
Tideman, N. 1995.
\newblock The Single Transferable Vote.
\newblock \emph{Journal of Economic Perspectives} 9(1): 27--38.

\bibitem[{Tideman and Richardson(2000)}]{TiRi00a}
Tideman, N.; and Richardson, D. 2000.
\newblock Better Voting Methods Through Technology: The
  Refinement-Manageability Trade-Off in the Single Transferable Vote.
\newblock \emph{Public Choice} 103(1-2): 13--34.

\bibitem[{Tideman(2006)}]{Tide06a}
Tideman, T.~N. 2006.
\newblock \emph{Collective Decisions And Voting: {T}he Potential for Public
  Choice}.
\newblock Ashgate.

\bibitem[{Woodall(1994)}]{Wood94a}
Woodall, D.~R. 1994.
\newblock Properties of preferential election rules.
\newblock \emph{Voting Matters} 3.

\bibitem[{Woodall(1997)}]{Wood97a}
Woodall, D.~R. 1997.
\newblock Monotonicity of single-seat preferential election rules.
\newblock \emph{Discrete Applied Mathematics} 77(1): 81--98.

\end{thebibliography}

\section*{Appendix}

			\subsection*{Proof of Proposition~\ref{prop: CPSC and IPSC differ}}

   %
   
   We prove the statement by two examples. 
   
   \begin{example}[IPSC does not imply CPSC.]
	   			First, we show that IPSC does not imply CPSC. Let $b_i = 1, L=2, C=\{a,b,c\}$ with $w(a)=w(c)=1$ and $w(b)=0.9$, and suppose that voters have dichotomous preferences:
	   			\begin{align*}
	   			1&: \quad  \{a,b\}, \\
	   			2&: \quad  \{a\}, \\
	   			3,4&: \quad  \{c\}.
	   			\end{align*}
	   		Consider the outcome $W=\{c,b\}$. This does not satisfy CPSC since the set of voters $N'=\{1,2\}$ is a generalised solid coalition for $C'=\{a\}$ with 
	   			$w(\{c\midd \exists i\in N' \text{ s.t }c \pref_i c^{(i, |C'|)} \}\cap W)=w(\{b\})=0.9<b(N')L/n=1,$
	   			and, yet, $C''=\{a\}\subseteq C'$ such that $w(C'')=1$. On the other hand, $W$ satisfies IPSC. For example, take $N'$ and $C'$ as above, there is a single candidate $a\in C'\backslash  \{c\midd \exists i\in N' \text{ s.t }c \pref_i c^{(i, |C'|)} \}\cap W$ and $w(\{a\}\cup \{c\midd \exists i\in N' \text{ s.t }c \pref_i c^{(i, |C'|)} \}\cap W)=w(\{a,b\})>1$. Thus, IPSC is not violated by the set of voters $N'$ and solid coalition $C'$. It can similarly be shown that for all other subsets of voters and sets of solidly supported candidates that IPSC is not violated. \hfill $\diamond$
	   \end{example}
	   
	   \begin{example}[CPSC does not imply IPSC]
		   			Second, we show that CPSC does not imply IPSC. Let $b_i = 1, L=2, C=\{a,b,c,d,y,z\}$ with $w(a)=w(y)=w(d)=2.1, w(b)=0.1, w(c)=0.9, w(z)=1.1$, and suppose that the voters' preferences are
		   			\begin{align*}
		   			1&: \quad  a, \{b,c\},z,d, y\\
		   			2&: \quad  b,\{a,d\},y,c,z \\
		   			3,4&: \quad  z, y, d, c, b,a   
		   			\end{align*}
		   		Consider  the outcome $W=\{c,z\}$. This does not satisfy IPSC. The set of voters $N'=\{1,2\}$ forms a generalised solid coalition for $C'=\{a,b\}$ and 
		   		$w(\{c\midd \exists i\in N' \text{ s.t }c \pref_i c^{(i, |C'|)} \}\cap W)=w(\{c\})=0.9<b(N')L/n=1.$
		   						However, the candidate $b\in C'\backslash \{c\midd \exists i\in N' \text{ s.t }c \pref_i c^{(i, |C'|)} \}\cap W$ and {$w(\{b,c\})=1$}. Thus, IPSC is violated. On the other hand, $W$ satisfies CPSC. For example, take $N'$ and $C'$ as above,  there is only one subset $C''=\{b\}\subseteq C'$  that does not exceed $b(N')L/n=1$. However, 
		   			$w(\{c\midd \exists i\in N' \text{ s.t }c \pref_i c^{(i, |C'|)} \}\cap W)=w(\{c\})=0.9 \ge w(\{b\})=0.1.$
		   			Thus, CPSC is not violated by the set of voters $N'$ and solid coalition $C'$. It can similarly be shown that for all other subsets of voters and sets of solidly supported candidates that CPSC is not violated. \hfill $\diamond$
		   \end{example}
	
   
\subsection*{Proof of Proposition~\ref{proposition: CPSC and IPSC exhaustive}}

   \begin{proof}
   Let $W$ be a non-exhaustive outcome. That is, there exists a candidate $c^*\in C\backslash W$ such that $w(W\cup \{c^*\})\le L$. 

   Later, in Proposition~\ref{proposition: CPSC maximally exhaustive}, we prove the stronger result that a CPSC outcome is always a maximal cost outcome. Thus, we omit the proof that a CPSC outcome is exhaustive. 


   For the sake of a contradiction, suppose that $W$ satisfies IPSC. The set of all voters $N$ solidly supports the entire candidate set $C$, $b(N)L/n=L$, and 
   $$ w(\{c\midd \exists i\in N' \text{ s.t }c \pref_i c^{(i, |C'|)} \}\cap W)=w(W)<L.$$
    Definition~\ref{definition: IPSC for PB} is violated since $c^* \in C\backslash (\{c\midd \exists i\in N' \text{ s.t }c \pref_i c^{(i, |C'|)}\}\cap W)$ and 
   $w(c^*\cup (\{c\midd \exists i\in N' \text{ s.t }c \pref_i c^{(i, |C'|)} \}\cap W))=w(W\cup \{c^*\})\leq b(N')L/n=L.$
   This is the desired contradiction.
   \end{proof}

   \subsection*{Proof of Proposition~\ref{proposition: CPSC maximally exhaustive}}
   

   \begin{proof}
   Suppose that $W$ and $W'$ are two distinct budgets that satisfy CPSC and assume $w(W)<w(W')\le L$. We prove that $W$ cannot satisfy CPSC. The set of all voters $N$ is a solid coalition for the entire candidate set $C$. Take $C''=W'$. We have $w(W')\le b(N')L/n=L$ and 
   $$w(\{c\midd \exists i\in N \text{ s.t }c \pref_i c^{(i, |C|)} \}\cap W)=w(W)< w(C'')=w(W').$$
   Thus, $W$ does not satisfy CPSC. 
   \end{proof}

    \subsection*{Proof of Proposition~\ref{prop: dichotomous pref CPSC}}



   \begin{proof}
   ($\Rightarrow$) We  prove the result via the contrapositive. Suppose that an outcome $W$ does not (simultaneously) satisfy (i) and (ii). If (ii) does not hold, then, by Proposition~\ref{proposition: CPSC maximally exhaustive}, CPSC does not hold.  Now, suppose that (ii) holds but (i) does not. That is,  $W$ is maximal cost and there exists $N'$ such that there exists $C''\subseteq \bigcap_{i\in N'} A_i$ with $w(C'')\le b(N')L/n$ and 
   \begin{align}\label{equation: CPSC for approvals0}
   w(W\cap \cup_{i\in N'} A_i)<w(C''). 
   \end{align}
   Since $C''\subseteq  \bigcap_{i\in N'} A_i$, the set of voters $N'$ forms a generalised solid coalition for $C''$ and 
   \begin{align}\label{equation: CPSC for approvals}
   \{c\midd c \pref_i c^{(i, |C''|)} \} &=A_i && \text{ for all $i\in N'$}.
   \end{align}
   Further, the (trivial) subset $C''$ is such that $w(C'')\le b(N')L/n$ and 
   $w(\{c\midd \exists i\in N' \text{ s.t }c \pref_i c^{(i, |C''|)} \}\cap W)=w(\cup_{i\in N'} A_i \cap  W)<w(C''),$
    by (\ref{equation: CPSC for approvals0}). Thus, CPSC is violated. 
   %
   %

   ($\Leftarrow$) We  prove the result via the contrapositive. Suppose that $W$ does not satisfy CPSC. If $W$ is not maximal cost, then (ii)  is violated and we are done. Assume that $W$ is  maximal cost but does not  satisfy CPSC. That is, $W$ is maximal cost and there exists  a set of voters $N'$ that solidly supports $C'$ such that $C''\subseteq C'$ and $w(C'')\le b(N')L/n$ but  
   \begin{align}\label{equation: CPSC for approvals 1}
   w(\{c\midd \exists i\in N' \text{ s.t }c \pref_i c^{(i, |C'|)} \}\cap W)<w(C'').
   \end{align}

   Now, suppose that, for some $i\in N'$,
   $$\{c\midd  c \pref_i c^{(i, |C'|)}\}\neq A_i.$$
   This can only occur if $|C'|>|A_i|$ or $|A_i|=0$. In both cases, this implies that $\{c\midd  c \pref_i c^{(i, |C'|)}\}=C$ and
   $$w(\{c\midd \exists i\in N' \text{ s.t }c \pref_i c^{(i, |C'|)} \}\cap W)=w(C\cap W)=w(W).$$
   But this is a contradiction since, combined with  (\ref{equation: CPSC for approvals 1}), this shows that $W$ cannot be a maximal cost outcome. Thus, it must be that, for all $i\in N'$,
   $$\{c\midd  c \pref_i c^{(i, |C'|)}\}= A_i,$$
   and the solidly supported candidate set $C'$ is a subset of $\bigcap_{i\in N'} A_i$. It follows that $C'' \subseteq C'$ is also a subset of   $\bigcap_{i\in N'} A_i$ such that $w(C'') \le b(N')L/n$ and 
   $$w(\{c\midd \exists i\in N' \text{ s.t }c \pref_i c^{(i, |C'|)} \}\cap W)=w(W\cap \cup_{i\in N'} A_i)<w(C''),$$
   by (\ref{equation: CPSC for approvals 1}). Thus, (ii) is violated; this completes the proof. 
   %
   %
   %
   %
   \end{proof}

   \subsection*{Proof of Proposition~\ref{prop: dichotomous pref IPSC}}


   \begin{proof}
   ($\Rightarrow$)   We  prove the result using the contrapositive. If (ii) does not hold, then, by Proposition~\ref{proposition: CPSC and IPSC exhaustive}, we see that IPSC is violated. Now, assume that (ii) holds but (i) does not hold. That is, $W$ is an exhaustive outcome, and there exists a set of voters $N'\subseteq N$ with
   \begin{align}\label{equation: approval IPSC eq1}
   w(\cup_{i\in N'}A_i \cap W)< b(N')L/n
   \end{align}
   and  some $c^*\in (\cap_{i\in N'}A_i)\setminus (\cup_{i\in N'}A_i\cap W)$ such that 
   \begin{align}\label{equation: approval IPSC eq2}
   w(\{c^*\} \cup (\cup_{i\in N'}A_i \cap W))\leq b(N')L/n.
   \end{align} 

   Let $C'=\bigcap_{i\in N'} A_i$. The set $C'$ is solidly supported by the set of voters $N'$ and, since $|C'|\le |A_i|$ for all $i\in N'$, we have 
   $$\{c \ : \ \ s.t. \ c\pref_i c^{(i, |C'|)}\}=A_i$$
   for all $i\in N'$. It then follows from (\ref{equation: approval IPSC eq1}) that
   $w(\{c \ : \ \exists i\in N' \ s.t. \ c\pref_i c^{(i, |C'|)}\}\cap W)=w(\cup_{i\in N'}A_i \cap W)<b(N')L/n,$
   and, yet, by (\ref{equation: approval IPSC eq2}) there exists a candidate 
   $c^*\in C'\backslash  \{c \ : \ \exists i\in N' \ s.t. \ c\pref_i c^{(i, |C'|)}\}\cap W=\cap_{i\in N'} A_i\backslash (\cup_{i\in N'}A_i \cap W)$
   such that 
   $w(c^*\cup \{c \ : \ \exists i\in N' \ s.t. \ c\pref_i c^{(i, |C'|)}\}\cap W)=w(c^*\cup \cup_{i\in N'}A_i \cap W)\le b(N')L/n.$
   That is, IPSC is violated. 

   %
   %
   %
   %
   %

   ($\Leftarrow$) We prove the result via the contrapositive. Suppose that  $W$ is an  outcome such that IPSC does not hold. If $W$ is not exhaustive, then (ii) is violated and we are done. Now, suppose the $W$ is exhaustive and does not satisfy IPSC. That is, $W$ is exhaustive, and there exists a set of voters $N'\subseteq N$ 
   who  solidly support a set of candidates $C'$ with 
   \begin{align}\label{equation: approval IPSC eq3}
   w(\{c\midd \exists i\in N' \text{ s.t }c \pref_i c^{(i, |C'|)} \}\cap W) <b(N')L/n
   \end{align}
    and there exists some candidate $c^*\in C' \setminus (\{c\midd \exists i\in N' \text{ s.t }c \pref_i c^{(i, |C'|)} \}\cap W)$ such that 
    \begin{align}\label{equation: approval IPSC eq4}
    w(c^*\cup (\{c\midd \exists i\in N' \text{ s.t }c \pref_i c^{(i, |C'|)} \}\cap W))\leq b(N')L/n.
    \end{align}
 
    First, suppose that, for some $i\in N'$, 
    $\{c \ : \ \ s.t. \ c\pref_i c^{(i, |C'|)}\}\neq A_i.$
   This can only occur if $|C'|>|A_i|$ or $|A_i|=0$. In either case, this implies that $\{c \ : \ \ s.t. \ c\pref_i c^{(i, |C'|)}\}=C$ and
   $ w(c^*\cup (\{c\midd \exists i\in N' \text{ s.t }c \pref_i c^{(i, |C'|)} \}\cap W))=w(c^*\cup W) \leq b(N')L/n\le L$
   for some $c^*\notin W$; but this contradicts the assumption that $W$ is exhaustive. Thus, it must be that 
     $\{c \ : \ \ s.t. \ c\pref_i c^{(i, |C'|)}\}= A_i,$
     for all $i\in N'$. It then follows that
   $$\{c\midd \exists i\in N' \text{ s.t }c \pref_i c^{(i, |C'|)} \}=\cup_{i\in N'} A_i,$$
   and, by (\ref{equation: approval IPSC eq3}), 
   $$w(\cup_{i\in N'} A_i \cap W)<b(N')L/n.$$
   Further, the candidate subset $C'$ must correspond to a subset of $\bigcup_{i\in N'} A_i$. Thus, the candidate $c^* \in (\cap_{i\in N'}A_i)\setminus (\cup_{i\in N'}A_i\cap W)$ and 
   $$w(\{c^*\}  \cup (\cup_{i\in N'}A_i\cap W))\le b(N')L/n.$$
    Thus, condition (i) is violated. 
   \end{proof}

    \subsection*{Proof of Proposition~\ref{prop:PBEAR-IPSC}}
   
		
   		\begin{proof}
   Let $W$ be an outcome of PB-EAR. For  sake of a contradiction, suppose that $W$ does not satisfy Inc-PSC. That is, there exists a set of voters $N'$ who solidly support a candidate set $C'$ such that 
   $$w(\bar{C}'\cap W)<b(N') L/n,$$
   where $\bar{C}':=\{c\midd \exists i\in N' \text{ s.t }c \pref_i c^{(i, |C'|)} \}$, and there exists a candidate $c^*\in C'\backslash (\bar{C'}\cap W)$ such that
   \begin{align}\label{equation: pb-ear incl-psc eq1}
   w\Big(c^*\cup (\bar{C}'\cap W)\Big)\le b(N')L/n.
   \end{align}
   We will denote $\bar{C}'\cap W$ by $W'$.

   First, suppose the PB-EAR terminated at some $j>|C'|$ iteration. At the end of the $j=|C'|$ iteration, the sum of voter weights in $N'$ is at least 
   $$b(N')-\sum_{c\in W'} w(c) n/L=b(N')- w(W') n /L.$$
   This follows because when each candidate is added to $W$ a total weight of $w(c) n /L$ is subtracted from the set of voters supporting this candidate. Our lower bound is attained by assuming that every candidate $c\in W$ that can possibly reduce the weight of voters in $N'$ (i.e., those candidates $c\in c^{(i, |C'|)}$ for $i\in N'$) subtracts the entire weight $w(c)/L$ from the voter set $N'$.

   But (\ref{equation: pb-ear incl-psc eq1}) implies that $w(W')\le  b(N')L/n-w(c^*)$, and so 
   $$b(N')-\sum_{c\in W'} w(c)n /L\ge w(c^*) n /L.$$
   Thus,  at the end of the  $j=|C'|$ iteration
   $$\sum_{i\in N \ : \ c^* \in A_i^{(j)}} b_i \ge n w(c^*)/L,$$
   which implies that $c^*\in C^*$ and $C^*\neq \emptyset$. This is contradiction since no other candidates from $\bar{C}'$ are contained in $W$ besides those already accounted for in $W'$, and so PB-EAR could not have iterated to the $j+1$-th stage. 

   Second, suppose that PB-EAR terminated at some $j\le |C'|$ iteration. This can only occur if $w(W)<L$ and no candidate can be added without exceeding the budget or $w(W)=L$.  The total voter weight that has been subtracted (from all voters)  via the algorithm is exactly
   \begin{align}\label{equation: pb-ear incl-psc eq2}
   \sum_{c\in W} w(c) n/L=w(W)n/L.
   \end{align}
   As noted in the above paragraphs, the voter weights of $N'$ is at least $n w(c^*)/L$ {and, hence, at most $n-n w(c^*)/L$ voter weight has been decreased from all voters. }   This gives an upper bound on the total voter weight that has been decreased from all voters {(\ref{equation: pb-ear incl-psc eq2}):}
   \begin{align}\label{equation: pb-ear incl-psc eq3}
  {w(W)n/L\le } n-n w(c^*)/L=[L-w(c^*)]n/L.
   \end{align} 
   {Simplifying  (\ref{equation: pb-ear incl-psc eq3}) gives}
   $$w(W)n/L \le [L-w(c^*)]n/L \implies w(W)\le L-w(c^*).$$
   That is, $w(W)+w(c^*)\le L$. This is a contradiction since the set $W$ does not equal $L$ and contains at least one candidate, namely $c^*$, such that $W\cup\{c^*\}$ does not exceed $L$. 
   			\end{proof}


			\subsection*{Proof of Proposition~\ref{proposition: CPSC equiv to genPSC and IPSC implies genPSC}}
		    	\begin{proof}
		   We begin with statement (i). Suppose $W$ does not satisfy CPSC.		
		    Then, there exists a set of candidates $C'$ solidly supported by $N'$,  
		    for which there is some  subset of candidates $C''\subseteq C'$ such that $|C''|\leq |N'|k/n$ but 
		    \begin{align}\label{equation: mw CPSC and genPSC eq}
		    |\{c\midd \exists i\in N' \text{ s.t }c \pref_i c^{(i, |C'|)} \}\cap W|< |C''|.
		    \end{align}

		    To show that generalised PSC does not hold, take $\ell=|C''|$, and notice that $N'$ is a generalised solid coalition for $C'$ with $|N'|\ge |C''|n/k=\ell n/k$. We  wish to show that there is no subset $C^*\subseteq W$ of size at least $\min\{\ell, |C'|\}=\ell$ such that for all $c''\in C^*$
		    $$\exists i\in N'\, : \quad c''\pref_i c^{(i, |C'|)}.$$ 
		    If such a $C^*$ set did exist, then  it must be that 
		    \begin{align*}
		    |\{c\midd \exists i\in N' \text{ s.t }c \pref_i c^{(i, |C'|)} \}\cap W| &\ge  |C^*|\ge \ell=|C''|,
		    \end{align*}
		    which contradicts (\ref{equation: mw CPSC and genPSC eq}). Therefore, no such set can exist and generalised PSC is violated. 

		    %
		    %

		    Suppose $W$ does not satisfy generalised PSC. Then, for some positive integer $\ell$, there exists a generalised solid coalition $N'$ supporting candidate subset $C'$ such  $|N'|\ge \ell n/k$  and there does not exist any subset $C^*\subseteq W \ : \ |C^*|\ge \min\{\ell, |C'|\}$  such that for all $c''\in C^*$
		    $$\exists i\in N'\, : \quad c''\pref_i c^{(i, |C'|)}.$$
		    This implies that  
		    \begin{align}\label{equation: mw CPSC and genPSC eq2}
		    |\{c\midd \exists i\in N' \text{ s.t }c \pref_i c^{(i, |C'|)} \}\cap W| < \min\{\ell, |C'|\}.
		    \end{align}

		    We now show that CPSC is violated. The set of voters $N'$ solidly supports $C'$. Let $C''\subseteq C'$ be any subset such that $|C''|=\min\{\ell, |C'|\}$. It follows that $|C''|\le \ell \le |N'|k/n$. However, from (\ref{equation: mw CPSC and genPSC eq2}) we have 
		    $$|\{c\midd \exists i\in N' \text{ s.t }c \pref_i c^{(i, |C'|)} \}\cap W| <\min\{\ell, |C'|\} =|C''|,$$
		    which is a violation of CPSC.
		    %
		    %
		    %
		    %
%
%

We now prove statement (ii). Suppose $W$ violates CPSC. Then, there exists a set of candidates $C'$ solidly supported by $N'$,  
				    for which there is a subset of candidates $C''\subseteq C'$ such that $|C''|\leq |N'|k/n$ but 
				    \begin{align}\label{equation: mw IPSC implies CPSC eq}
				    |\{c\midd \exists i\in N' \text{ s.t }c \pref_i c^{(i, |C'|)} \}\cap W|< |C''|.
				    \end{align} In the multi-winner setting, $w(c)=1$ for all $c\in C$ and so 
				    $$|\{c\midd \exists i\in N' \text{ s.t }c \pref_i c^{(i, |C'|)} \}\cap W|\le |C''|-1.$$
				    Now let $c^*$ be some candidate $c^*\in C'' \setminus (\{c\midd \exists i\in N' \text{ s.t }c \pref_i c^{(i, |C'|)} \}\cap W)$, such a candidate must exist by (\ref{equation: mw IPSC implies CPSC eq}). But then
				    $$|c^*\cup (\{c\midd \exists i\in N' \text{ s.t }c \pref_i c^{(i, |C'|)} \}\cap W)|\le |C''|\leq |N'|k/n,$$ 
				    and IPSC is violated. Therefore in multi-winner voting, IPSC implies Generalised PSC (or CPSC).  
				    		\end{proof}

	    \subsection*{Proof of Proposition~\ref{prop: IPSC implies PJR for approval-based mw and EJR}}

	
   	\begin{proof}
	Statement (i) follows immediately from statement (ii) in Proposition~\ref{proposition: CPSC equiv to genPSC and IPSC implies genPSC}  and Corollary~\ref{corollary: approval voting CPSC and PJR equiv}.
	
%

   		For statement (ii), we show that PJR does not imply IPSC. 
   		\begin{align*}
   			1-3:&\quad \{a,x\}\\
   			4-6:&\quad \{a,y\}\\
   			7-12: &\quad \{u,v,w,x,y,z\}
   			\end{align*}
		
   			Consider outcome $W=\{u,v,w,x,y,z\}$ for $k=6$. Then consider the set $N'=\{1,2,3,4,5,6\}$. PJR is not violated and hence $W$ satisfies PJR. However, IPSC is violated and so $W$ is not IPSC.   

		
%
%
%
%
		  			Finally, we prove statement (iii). We begin by showing that an EJR committee need not be IPSC. This follows from the example given in the  statement (ii) 
					where the outcome $W$ satisfies EJR but does not satisfy IPSC.

		   Second, we show that an IPSC committee need not be EJR. Consider the following example where $n=6, k=3$, and voter preferences are as below.
				\begin{align*}
				1-2:&\quad \{a,b,c\}\\
				3-4:&\quad \{a,b,d\}\\
				5-6:&\quad \{z\}.
				\end{align*}
				The outcome $W=\{c,d,z\}$ does not satisfy EJR. This follows since $N'=\{1,2,3,4\}$ is such that $|N'|\ge 2 (n/k)=4$ and $|\cap_{i\in N'}A_i|\ge 2$, yet $|A_i\cap W|=1<2$ for each $i\in N'$. On the other hand, $W$ does satisfy IPSC. This is because the solid coalition $N'$ is not sufficiently large enough to meet condition (i) of the definition since $ (|\cup_{i\in N'}A_i \cap W|+1 ) n/k=6>4=|N'|$.
		   \end{proof}

	    \subsection*{Proof of Proposition~\ref{prop:pav-implies-ipsc}}

		
		   		\begin{proof}
		   Let $W$ be the output of PAV. Suppose for the sake of a contradiction that there exists $N'$ satisfying (i) and (ii).

		   			For each candidate $w\in W$, we define its marginal contribution as the difference between the PAV score of $W$ and $W\backslash \{w\}$. Let $m(W)$ be the sum of marginal differences of all candidates in $W$. Note that if $c^*$ were added to $W$, then the PAV score would increase by at least
		   		$$|N'|\frac{1}{|\cup_{i\in N'}A_i \cap W|+1 }\ge n/k.$$
		   		Thus, it suffices to prove that the marginal contribution of some candidate in $W$ is less than $n/k$.

		   		Consider the set $N\backslash N'$. We have $|N\backslash N'|\le n-n/k=n/k(k-1)$. Pick a voter $i\in N\backslash N'$ and let $j=|A_i\cap W|$. If $j>0$, then this voter contributes exactly $1/j$ to the marginal contribution of each candidate in $A_i\cap W$, and hence her contribution to $m(W)$ is exactly 1. If $j=0$, this voter does not contribute to $m(W)$ at all. Therefore, we have that $m(W)\le |N\backslash N'|\le n/k(k-1)$. Applying the pigeonhole principle, we see that there exists some candidate $w\in W$ with marginal contribution less than $n/k$. This complete the proof.
		   		\end{proof}

		   	    \subsection*{Proof of Proposition~\ref{Proposition: IPSC and EJR via PAV}}


\begin{proof}
The proof of statement (i) follows immediately from Proposition~\ref{prop:pav-implies-ipsc} and the fact that PAV also implies EJR~\citep{ABC+16a}.

We now prove statement (ii). Consider the example from  the proof of statement (ii) in Proposition~\ref{prop: IPSC implies PJR for approval-based mw and EJR}; we will show that the outcome $W'=\{a,u,v,w,x,z\}$ satisfies both EJR and IPSC but is not PAV. 

We begin with EJR. Note that the solid coalition $N'=\{1, 2, \ldots, 6\}$  has size 6 and their approval sets  have an intersection of size 1. Yet, the each voter in $N'$ has at least one of their approved candidates elected.  IPSC is also satisfied. Again consider $N'$, these voters have size 6 and are sufficiently large since $6\ge(|\cup_{i\in N'}A_i \cap W'|+1 ) n/k$; however, there does not exist any candidate in their intersection that is not already elected. Thus both EJR and IPSC are satisfied by $W'$. 

It only remains to prove that $W'$ is not PAV. This is straightforward. The PAV score of $W'$ is 
$$3(1+1/2)+3+6(1+1/2+\cdots+1/5)=21.2.$$
This cannot be the PAV outcome since $W^*=\{a,v,w,x,y,z\}$ has PAV score of 
$$3(1+1/2)+3(1+1/2)+6(1+1/2+\cdots+1/5)=22.7.$$
We note, as an aside, that it is, of course, true that $W^*$ also satisfies both  EJR and IPSC (as per statement (i)).
\end{proof}

\end{document}